\documentclass[preprint,10pt]{mytest}
\usepackage{url}
\usepackage{amsmath}
\usepackage{amssymb}
\usepackage{amsthm}
\usepackage{graphics}
\usepackage{graphicx}
\usepackage{times}
\usepackage{algorithm,algorithmic}

\newcommand{\EBeta}[2]{\mathrm{B}(#1,#2)}

\newtheorem{theorem}{Theorem}

\newtheorem{Definition}[theorem]{Definition}
\newtheorem{Lemma}[theorem]{Lemma}

\newtheorem{Claim}[theorem]{Claim}

\urldef{\mailsa}\path|rafal.kapelko@pwr.edu.pl|

\begin{document}

\begin{frontmatter}
%\title{This is a specimen title\tnoteref{t1,t2}}
%\tnotetext[t1]{This document is a collaborative effort.}
\title{On the 
Displacement for Covering a $d-$dimensional Cube with Randomly Placed Sensors\tnoteref{t1}}
%\fnref{awardfootnote}
%\corref{cor2}
\tnotetext[t1]{This is an expanded and revised version of a paper that appeared in \cite{KapelkoK15} and received the Best Paper Award at the 14th International Conference, ADHOC-NOW 2015.}
%On Displacement to the Power of Random Sensors for Coverage in 1D - draft version
%\thanks{Supported by grant nr XXX} %2012/S1????? of the Institute of Mathematics and Computers Science of the Wroc{\l}aw University of Technology} 
%\author{}
\author[pwr]{Rafa\l{} Kapelko\corref{cor1}\fnref{pwrfootnote}}
\ead{rafal.kapelko@pwr.edu.pl}
\author[scs]{Evangelos Kranakis\fnref{scsfootnote}}%\corref{cor2}}
\ead{kranakis@scs.carleton.ca }
%\fntext[awardfootnote]{This is an expanded and revised version of a paper that appeared in \cite{KapelkoK15} and received the Best Paper Award at the 14th International Conference, 
%ADHOC-NOW 2015.}
\fntext[pwrfootnote]{Research supported by grant nr S40012/K1102}
\fntext[scsfootnote]{Research supported in part by NSERC Discovery grant.}
\cortext[cor1]{Corresponding author at: Department of Computer Science,
Faculty of Fundamental Problems of Technology, Wroc{\l}aw University of Technology, 
 Wybrze\.{z}e Wyspia\'{n}skiego 27, 50-370 Wroc\l{}aw, Poland. Tel.: +48 71 320 33 62; fax: +48 71 320 07 51.}
 %\cortext[cor2]{Corresponding author at: School of Computer Science, Carleton University
%1125 Col. By Dr., Ottawa, ON K1S 5B6, Canada. Tel.: 613 520 2600 x 8090; fax:a 613 520 4334.}
\address[pwr]{Department of Computer Science,Faculty of Fundamental Problems of Technology, Wroc{\l}aw University of Technology, Poland}
\address[scs]{School of Computer Science, Carleton University, Ottawa, ON, Canada}

\begin{abstract}
Consider $n$ sensors placed randomly and independently with the uniform distribution
in a $d-$dimensional unit cube ($d\ge 2$).  
The sensors have identical sensing range equal to $r$, for some $r >0$. We are interested in moving the sensors from their initial positions to new positions so as to ensure that the $d-$dimensional unit cube  is completely covered, i.e., 
every point in the $d-$dimensional cube is within the range of a sensor.  If the $i$-th sensor is displaced a distance $d_i$, what is a displacement of minimum cost? As cost measure for the displacement of the team of sensors we consider the {\em $a$-total movement} defined as the sum $M_a:= \sum_{i=1}^n d_i^a$, for some constant $a>0$. We assume that $r$ and $n$ are chosen so as to allow full coverage 
of the $d-$dimensional unit cube and $a > 0$. 

The main contribution of the paper is to show the existence of a tradeoff between the $d-$dimensional cube, sensing radius and $a$-total movement. The main results can be summarized as follows for the case of the $d-$dimensional cube.
\begin{enumerate}
\item
If the  $d-$dimensional cube sensing radius is $\frac{1}{2n^{1/d}}$ and $n=m^d$, for some $m\in N$,
then we present an algorithm that uses
$O\left(n^{1-\frac{a}{2d}}\right)$ total expected movement 
(see Algorithm \ref{alg_3_4} and Theorem \ref{thm:3_4}).
\item 
If the
the $d-$dimensional cube 
sensing radius is greater than 
$\frac{3^{3/d}}{(3^{1/d}-1)(3^{1/d}-1)}\frac{1}{2n^{1/d}}$ and $n$ is a natural number
then the total expected movement is 
$O\left(n^{1-\frac{a}{2d}}\left(\frac{\ln n}{n}\right)^{\frac{a}{2d}}\right)$ 
(see Algorithm \ref{alg_log} and Theorem \ref{thm:log}).
\end{enumerate}
This sharp decline from $O(n^{1- \frac{a}{2d}})$ to $O(n^{1-\frac{a}{2d}} (\ln n)^{\frac{a}{2d}} )$
in the $a$-total movement of the sensors to attain complete coverage of the $d-$dimensional cube
indicates the presence of an interesting threshold on the sensing radius in a $d-$dimensional cube as it increases from $\frac{1}{2n^{1/d}}$ to 
$ \frac{3^{3/d}}{(3^{1/d}-1)(3^{1/d}-1)}\frac{1}{2n^{1/d}}$.
In addition, we simulate Algorithm \ref{alg_3_4} and discuss the results of our simulations.

\end{abstract}


\end{frontmatter}

\section{Introduction}
\label{sec:Introduction}

A key challenge in utilizing effectively a group of sensors is to make them form 
an interconnected structure with good communication characteristics. For example, one may want to establish a sensing and communication infrastructure for robust connectivity, surveillance, security, or even reconnaissance of an urban environment using a limited number of sensors. For a team of sensors initially placed in a geometric domain such a robust connectivity cannot be assured a priori e.g., due to geographic obstacles (inhibiting transmissions), harsh environmental conditions (affecting signals), sensor faults (due to misplacement), etc. In those cases it may be required that a group of sensors originally placed in a domain be displaced to new positions either by a centralized or distributed controller. The main question arising is what is the cost of displacement so as to move the sensors from their original positions to new positions so as to attain the desired communication characteristics?

A typical sensor is able to sense a limited region usually defined by its sensing radius, say $r$, and considered to be a circular domain (disc of radius $r$). To protect a larger region against intruders every point of the region must be within the sensing range of at least one of the sensors in the group. Moreover, by forming a communication network with these sensors one is able to transmit to the entire region any disturbance that may have occurred in any part of the region. This approach has been previously studied in several papers. It includes research on 1) {\em area coverage} in which one ensures monitoring of an entire region \cite{full_coverage2,full_coverage1}, and 2) on {\em perimeter} or {\em barrier coverage} whereby a region is protected by monitoring its perimeter thus sensing intrusions or withdrawals to/from the interior \cite{barriercoverageNodeDegree,Evan09,MinMax,barrierMinSum,barriercoverage05}.
Note that barrier coverage is less expensive (in terms of number of sensors) than area coverage. Nevertheless, barrier coverage can be only used to monitor intruders to the area, as opposed to area coverage that can also protect the interior.

\subsection{Related work}
Assume that $n$ sensors of identical range are all initially placed on a line. It was shown in \cite{MinMax} that there is an $O(n^2)$ algorithm for minimizing the max displacement of a sensor while the optimization problem becomes NP-complete if there are two separate (non-overlapping) barriers on the line (cf. also \cite{che2} for arbitrary sensor ranges). If the optimization cost is the sum of displacements then \cite{barrierMinSum} shows that the problem is NP-complete when arbitrary sensor ranges are allowed, while an $O(n^2)$ algorithm is given when all sensing ranges are the same.  Similarly, if one is interested in the number of sensors moved then the coverage problem is NP-complete when arbitrary sensor ranges are allowed, and an $O(n^2)$ algorithm is given when all sensing ranges are the same \cite{mon}.  Further, \cite{dobrev2013complexity} considers the algorithmic complexity of several natural generalizations of the barrier coverage problem with sensors of arbitrary ranges, including when the initial positions of sensors are arbitrary 
points in the two-dimensional plane, as well as multiple barriers that are parallel or perpendicular to each other.

An important setting in considerations for barrier coverage is when the sensors are placed at random on the barrier according to the uniform distribution. Clearly, when the sensor dispersal on the barrier is random then coverage depends on the sensor density and some authors have proposed using several rounds of random dispersal for complete barrier coverage \cite{eftekhari13,yan}. Another approach is to have the sensors relocate from their initial position to a new position on the barrier so as to achieve complete coverage 
\cite{MinMax,barrierMinSum,eftekhari13d,mon}. Further, this relocation may be done in a centralized (cf. \cite{MinMax,barrierMinSum}) or distributed manner (cf. \cite{eftekhari13d}).

Closely related to our work is \cite{spa_2013}, where
algorithm $MV_1(n,y)$ was analysed. In this paper, $n$ sensors were placed in the unit interval uniformly and independently at random and the cost of displacement was measured by the sum of the respective displacements of the individual sensors  in the unit line segment $[0,1].$
Lets call the positions $\frac{i}{n} - \frac{1}{2n}$, for $i = 1, 2, \ldots , n$, {\em anchor} positions. The sensors have the sensing radius $r=\frac{1}{2n}$ each.
Notice that the only way to attain complete coverage is for the sensors to occupy the anchor positions.
The following result was proved in \cite{spa_2013}.
\begin{theorem}[cf. \cite{spa_2013}]
\label{thm:spaa}
%\textit{
Assume that, $n$ mobile sensors are thrown uniformly and independently at random in the unit interval. 
The expected sum of displacements of all $n$ sensors to move from their current location to the equidistant anchor 
locations $\frac{i}{n} - \frac{1}{2n}$, for $i = 1, 2, \ldots , n$,  respectively, is in $\Theta(\sqrt{n}).$
%}
\end{theorem}
\begin{algorithm}
\caption{$MV_1(n,y)\,\,$ (Sensor displacement on a  interval).}
\label{alg_interval}
\begin{algorithmic}[1]
 \REQUIRE $n$ mobile sensors with identical sensing radius $r=\frac{y}{2n}$ placed uniformly and independently at random on the interval $[0,y].$
 \ENSURE  The final positions of the sensors are at the locations $\left(\frac{yi}{n}-\frac{y}{2n}\right),$
 $1\le i\le n$ (so as to attain coverage of the interval $[0,y].$)
 \STATE{sort the initial locations of sensors; the locations after sorting $x_1,x_2,\dots x_n,$  $x_1\le x_2\le\dots \le x_n.$}
 \FOR{$i=1$  \TO $n$ } 
% \STATE{ 
 %\FOR{$i=1$  \TO $\sqrt{n}$ } 
 \STATE{move the sensor $S_{i}$ at position $\left(\frac{yi}{n}-\frac{y}{2n}\right)$}
 %\ENDFOR
% } 
 \ENDFOR
 %\FOR{$j=1$  \TO $\sqrt{n}$ } 
 %\STATE{move sensors $S_{1+(j-1)\sqrt{n}}, S_{2+(j-1)\sqrt{n}},\dots,S_{\sqrt{n}+(j-1)\sqrt{n}}$ to equidistant points that are sufficient to cover the
 %rectangle with vertices
 %$
 %\left(\frac{j-1}{\sqrt{n}},0\right),$ $\left(\frac{j}{\sqrt{n}},0\right),$ $\left(\frac{j-1}{\sqrt{n}},1\right),
 %$ 
 %$\left(\frac{j}{\sqrt{n}},1\right)$ }
 %\ENDFOR
\end{algorithmic}
\end{algorithm}

In \cite{kkpower}, Theorem~\ref{thm:spaa} was extended to when the
cost of displacement is measured by the sum of the respective displacements 
raised to the power $a>0$ of the respective sensors in the unit line segment $[0,1].$
The following result was proved.
\begin{theorem}[cf. \cite{kkpower}]
\label{them:kkpower} 
Assume that $n$ mobile sensors are thrown uniformly and independently at random in the unit interval. 
The expected sum of displacements to a given power $a$ of algorithm $MV_1(n,1)$
%all $n$ sensors to move from their current location to the respective equidistant anchor positions $\frac{i}{n} - \frac{1}{2n}$, for $i = 1, 2, \ldots , n$, 
is in
$\Theta\left(1 / n^{\frac{a}{2}-1}\right),$ when $a$ is natural number, and in $O\left(1 /n^{\frac{a}{2}-1}\right)$, when $a>0.$
\end{theorem}

An analysis similar to the one for the line segment was provided for the unit square in \cite{spa_2013}. 
Our present paper focuses on the analysis of sensor displacement for a group of sensors placed uniformly at random on the $d-$dimensional unit cube, thus also generalizing the results of \cite{KapelkoK15} from $d=2$ to arbitrary dimension $d \geq 2$. 
In particular, our approach is the first to generalize the results of \cite{spa_2013} to the $d-$dimensional unit cube using as cost metric the $a$-total movement, 
and also obtain sharper bounds for the case of the unit square.

\subsection{Preliminaries and notation}
Let $d$ be a natural number.
We define below the concept {\em d-Dimensional Cube Sensing Radius}
which refers to a coverage area having the shape of a $d$-dimensional cube.\footnote{Recall that the generally accepted coverage area of a sensor is a $d$-dimensional ball. 
Our results can be easily converted to this model by describing a minimum d-dimensional ball outside this $d$-dimensional cube.}

\label{subsec:preliminar}
\begin{Definition}[d-Dimensional Cube Sensing Radius]
Consider a sensor\\ 
$Z_{(x_1,x_2,\dots,x_d)}$ located in position $(x_1,x_2,\dots,x_d), $ where $0 \le x_1,x_2,\dots,x_d\le 1.$ 
We define the range of the sensor $Z_{(x_1,x_2,\dots, x_d)}$
to be the area delimited by the d-dimensional cube with the $2^d$ vertices
$(x_1\pm r,x_2\pm r,\dots x_d\pm r)$, and call $r$ the d-dimensional cube sensing radius of the sensor.
%$ $(x-a,y+a),$ $(x+a,y-a),$ $(x-a,y-a).$, 
%Then the square sensing radius of the sensor
%$Z$ is $a.$
\end{Definition}

We also define the cost measure {\em $a$-total movement} as follows. 
\begin{Definition}[$a$-total movement]
Let $a>0$ be a constant. Suppose the displacement of the $i$-th sensor is a distance $d_i$. The {\em $a$-total movement} is defined as the sum $M_a:= \sum_{i=1}^n d_i^a$. 
(We assume that, $r$ and $n$ are chosen so as to allow full coverage of the $d$-dimensional cube and $a>0$.)
\end{Definition}
Motivation for using this cost metric arises from the fact that there might be a terrain with obstacles that obstruct the sensor movement from their initial to their final destinations. 
Therefore the $a$-total movement is a more realistic metric than the one previously considered for $a=1$.

In the analysis below we consider the Beta distribution. 
We say that a random variable concentrated on the interval $[0,1]$ has the
$\EBeta{a}{b}$ distribution with parameters $a,b,$ if it has the probability density function
\begin{equation}
\label{eq:beta1}
f(x)=\frac{1}{\EBeta{a}{b}}x^{a-1}(1-x)^{b-1} ,
\end{equation}
where the Euler Beta function (see \cite{NIST})
\begin{equation}
\label{eq:beta2}
\EBeta{a}{b} = \int_0^1 x^{a-1}(1-x)^{b-1} dx
\end{equation}
is defined for all complex numbers $a,b$ 
such as $\Re(a)>0$ and $\Re(b)>0$. 
%We will also use the following basic identity
%$\EBeta{a}{b} = \Gamma(a)\Gamma(b)/\Gamma(a+b)$.
Let us notice that for any integer numbers $a,b\ge 0,$ we have
\begin{equation}
\label{Emult}
\EBeta{a}{b}^{-1}=\binom{a+b-1}{a}a.
\end{equation}

%\subsection{Related Results}
%\label{subsec:related}
\subsection{Results of the paper}
\label{subsec:contribution}

We consider $n$ mobile sensors with identical $d-$dimensional cube  sensing radius
$r$ placed independently at random with the uniform distribution in the $d-$dimensional unit cube $(d\ge 2).$  
We want to have the sensors move from their current location to positions that cover the  $d-$dimensional cube in the sense that every point in the $d-$dimensional cube
is within the range of at least one sensor.
When a sensor is displaced on the $d-$dimensional cube a distance equal to $d$  the cost of 
the displacement is $d^a$ for some (fixed) power
$a>0$ of the distance $d$ traveled. We assume that $r$ and $n$ are chosen so as to allow full coverage of the $d-$dimensional cube,
i.e., every point of the region
is within the range of at least one sensor.

The main contribution of the paper in Section~\ref{sec:2d} is to show the existence of a trade off between $d-$dimensional cube  sensing radius and $a$-total movement 
that can be summarized as follows: 
%The main contributions of the paper can be summarized as follows:
\begin{enumerate}
%In this section we recall some known facts about special functions and special numbers which will be useful in the analysis
%in the next sections. 
\item
For the case of the $d-$dimensional cube  sensing radius $\frac{1}{2n^{1/d}}$ and $n=m^d$ for some $m\in N$
we present an algorithm that uses
$O\left(n^{1-\frac{a}{2d}}\right)$ total expected movement 
(see Algorithm \ref{alg_3_4} and Theorem \ref{thm:3_4}).
\item
If the $d-$dimensional cube sensing radius is greater than
$\frac{3^{3/d}}{(3^{1/d}-1)(3^{1/d}-1)}\frac{1}{2n^{1/d}}$ and $n$ is a natural number
then the expected movement is 
$O\left(n^{1-\frac{a}{2d}}\left(\frac{\ln n}{n}\right)^{\frac{a}{2d}}\right)$ 
(see Algorithm \ref{alg_log} and Theorem \ref{thm:log}).
\end{enumerate}
Notice that, for $a=d$ Algorithm $MV_d(n,1)$ uses 
$O(\sqrt{n})$ total expected movement while Algorithm $LV_d(n)$
uses $O(\sqrt{\ln n})$ total expected movement.
Therefore this sharp decrease from $O(n^{1- \frac{a}{2d}})$ to $O(n^{1-\frac{a}{2d}} (\ln n)^{\frac{a}{2d}} )$
in the $a$-total movement of the sensors to attain complete coverage of the $d-$dimensional cube
indicates the presence of an interesting threshold on the $d-$dimensional cube sensing radius when it increases from $\frac{1}{2n^{1/d}}$ to 
$ \frac{3^{3/d}}{(3^{1/d}-1)(3^{1/d}-1)}\frac{1}{2n^{1/d}}$.

In Section~\ref{sec:simu} we simulate Algorithm \ref{alg_3_4} and provide the results of the simulations. Finally, Section~\ref{sec:conclusion} is the conclusion.
\section{Displacement in $d-$dimensional cube}
\label{sec:2d}

%To simplify proofs we will assume without loss of generality that the number $n$
%is itself the square of a natural number and let $0<a\le 4$.
Assume that $n$ mobile sensors with the same $d-$dimensional cube  sensing radius are thrown uniformly
at random and independently in the $d-$dimensional unit cube $[0,1]^d$. Let $a>0$ and $d\ge 2.$

Our first result is an upper bound on the expected $a-$total movement
for the case, where the $d-$dimensional cube sensing radius
is $\frac{1}{2n^{1/d}}.$ Observe that in this case
the only way for the sensors to attain complete coverage
of the $d-$dimensional unit cube is to occupy the positions
$$
\left(\frac{l_1}{n^{1/d}}-\frac{1}{2n^{1/d}},
\frac{l_2}{n^{1/d}}-\frac{1}{2n^{1/d}},
\dots ,\frac{l_d}{n^{1/d}}-\frac{1}{2n^{1/d}}
\right),
$$
for $1\le l_1,l_2,\dots , l_d\le n^{1/d}$ and $l_1,l_2,\dots,l_d\in N.$ Let us also notice
that $n=m^d$ for some $m\in N.$
We present a recursive algorithm $MV_d(n,1)$ that uses 
$O\left(n^{1-\frac{a}{2d}}\right)$ expected $a-$total movement.

The algorithm is in two-phases. 
During the first phase (see steps $(1-6)$) we apply a greedy strategy 
and move all the sensors
only according to the first coordinate. 
Figure \ref{fig:idea} illustrates the steps $(1-6)$ of Algorithm $MV_2(n,1)$
\begin{figure}
  \begin{center}
    \includegraphics[width=0.75\textwidth]{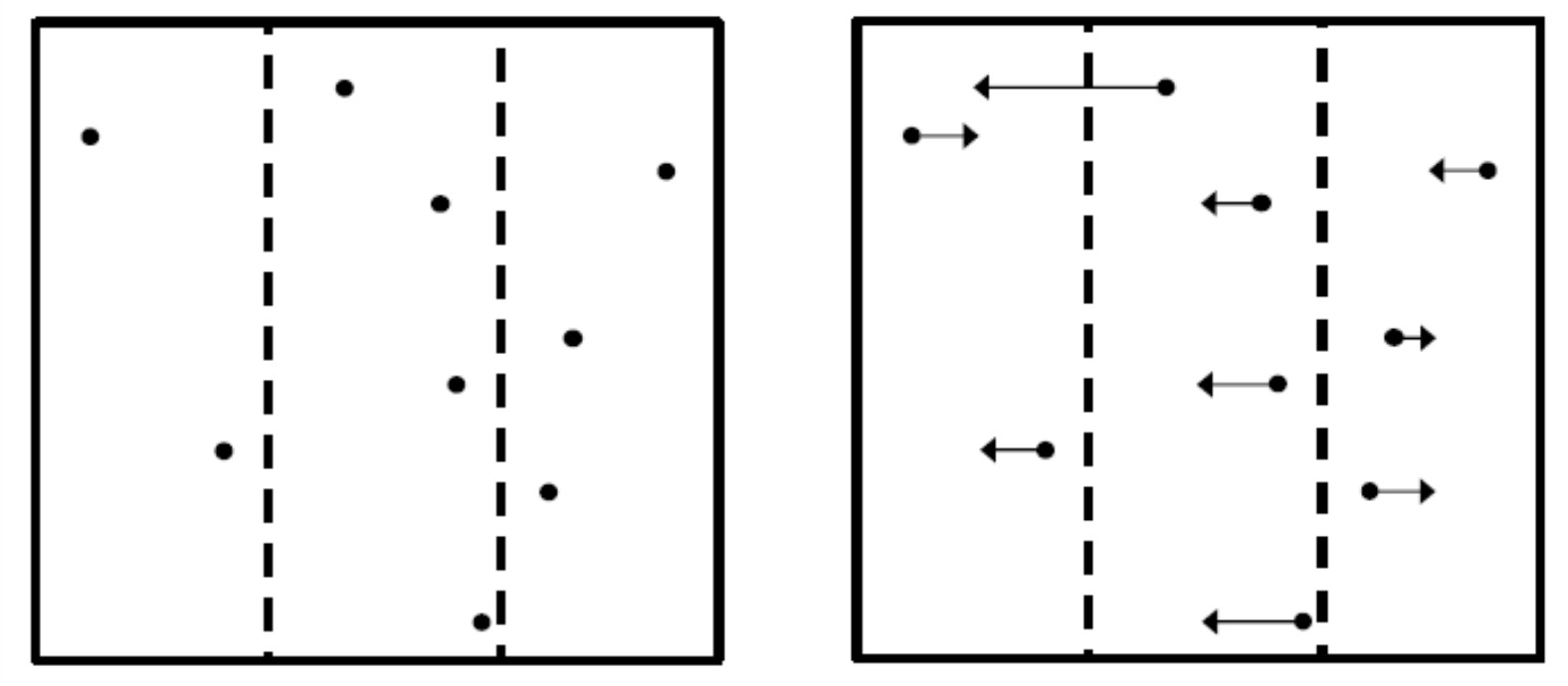}
  \end{center}
  \caption{Nine mobile sensors located in the interrior of a unit square move to new positions according to steps $(1-6)$ of Algorithm $MV_2(9,1)$.}
    \label{fig:idea}
\end{figure}
As a result of the first phase we have $n^{1/d}$ $(d-1)$-dimensional cubes each with $n^{\frac{d-1}{d}}$ random sensors. 
Hence the first phase reduces the sensor movement on the unit $d$-dimensional cube to the sensor movement on the unit $(d-1)$-dimensional cube.
During the second phase (see steps $(7-9)$) we move sensors in the unit  $(d-1)$-dimensional cube. Notice that for the base case $d=1$ we execute
algorithm $MV_1(n,1).$
%Theorem \ref{them:kkpower} [cf. \cite{kkpower}].  
\begin{algorithm}
\caption{$MV_d(n,y)\,\,$ Sensor displacement on a $d$-dimensional cube  when $n$ is $d$th power of the natural number and $d\ge 2.$}
\label{alg_3_4}
\begin{algorithmic}[1]
 \REQUIRE $n$ mobile sensors with identical $d$-dimensional cube sensing radius $r=\frac{1}{2n^{1/d}}$ placed uniformly and independently at random on 
 the $d-$dimensional cube  $[0,y]^d$
 \ENSURE  The final positions of the sensors are at the locations $\left(\frac{yl_1}{n^{1/d}}-\frac{y}{2n^{1/d}},\frac{yl_2}{n^{1/d}}-\frac{yl_2}{2n^{1/d}},\dots ,
 \frac{yl_d}{n^{1/d}}-\frac{yl_d}{2n^{1/d}}\right),$
 $1\le l_1,l_2,\dots , l_d\le n^{1/d}$\\ and $l_1,l_2,\dots, l_d\in N$ (so as to attain coverage of $[0,y]^d$)
 \STATE{Sort the initial locations of sensors according to the first coordinate; the locations \\
 after sorting $S_1=(x_1(1), x_2(1), \dots , x_d(1)), 
 S_2=(x_1(2), x_2(2), \dots ,x_d(2)),\dots$ \\
 $S_n=(x_1(n), x_2(n), \dots , x_d(n)),\,\,\,$  $x_1(1)\le x_1(2)\le\dots \le x_1(n);$}
 \FOR{$j=1$  \TO $n^{1/d}$ } 
% \STATE{ 
 \FOR{$i=1$  \TO $n^{(d-1)/d}$ } 
 \STATE{Move the sensor $S_{i+(j-1)n^{1/d}}$ at position\\ 
 $\left(\frac{jy}{n^{1/d}}-\frac{y}{2n^{1/d}},x_{i+(j-1)n^{1/d}}(2),\dots ,
 x_{i+(j-1)n^{1/d}}(d)\right);$}
 \ENDFOR
% } 
 \ENDFOR
 \FOR{$j=1$  \TO $n^{1/d}$ } 
 \STATE{Execute $MV_{d-1}(n^{(d-1)/d},y)$ for sensors\\ 
 $S_{1+(j-1)n^{1/d}}, S_{2+(j-1)n^{1/d}},\dots,S_{n^{(d-1)/d}+(j-1)n^{1/d}};$} 
 \ENDFOR
\end{algorithmic}
\end{algorithm}

We prove the following theorem.
\begin{theorem}
\label{thm:3_4}
Fix $d\in N.$
Let $n=m^d$ for some $m\in N$ and  let $a>0$.
Assume that $n$ sensors of d-dimensional cube sensing 
radius equal to $\frac{1}{2n^{1/d}}$ are thrown randomly and uniformly and 
independently with the uniform distribution on a unit d-dimensional cube. 
The expected $a-$total movement of algorithm $MV_d(n,1)$
is in $O\left(n^{1-\frac{a}{2d}}\right).$
\end{theorem}
\begin{proof}
We will prove the statement of the theorem by mathematical induction. Observe that the base case for $d=1$ follows from 
Theorem \ref{them:kkpower} [cf. \cite{kkpower}].  
%We will present only the idea of proof without all technical details.
%Without loss of generality let $n$ sensors be at the initial locations $(x_1, y_1),$ $(x_2,y_2),$ $\dots (x_n,y_n)$ such that $x_1<x_2<\dots x_n.$
Let us assume the result holds for the number $d-1.$ Let $a>0.$
We will estimate the expected $a-$total movement at the steps $(1-6).$ 
Let $X_i$ be the $i$th order statistic, i.e., the position of the $i$th sensor in the interval $[0,1]$ after sorting in step $(1)$. 
It turns out (see \cite{nagaraja_1992}) that $X_i$ obeys the Beta distribution with parameters $i, n-i+1$.
We know that the density function for $X_i$ (see Equation~(\ref{eq:beta1})) is
$$f_{X_i}(x)=i\binom{n}{i}x^{i-1}(1-x)^{n-i}.$$
Therefore, the  expected $a-$total movement in steps $(1-6)$
of algorithm $MV_d(n,1)$
is equal to
$$E_{(1-6)}^{(a)}=\sum_{j=1}^{n^{1/d}}\sum_{i=(j-1)n^{(d-1)/d}+1}^{jn^{(d-1)/d}}i\binom{n}{i}\int_{0}^{1}\left|x-\left(\frac{j}{n^{1/d}}-\frac{1}{2n^{1/d}}\right)\right|^a
x^{i-1}(1-x)^{n-i}dx.$$
Notice that, the expected $a-$total movement of algorithm $MV_1(n,1)$
is equal to
\begin{align*}
D^{(a)}&=\sum_{i=1}^{n}
i\binom{n}{i}\int_0^1\left|x-\left(\frac{i}{n}-\frac{1}{2n}\right)\right|^ax^{i-1}(1-x)^{n-i}dx\\
&=
%\sum_{j=1}^{\sqrt{n}}\sum_{i=(j-1)\sqrt{n}+1}^{j\sqrt{n}}
\sum_{j=1}^{n^{1/d}}\sum_{i=(j-1)n^{(d-1)/d}+1}^{jn^{(d-1)/d}}
i\binom{n}{i}\int_0^1\left|x-\left(\frac{i}{n}-\frac{1}{2n}\right)\right|^ax^{i-1}(1-x)^{n-i}dx.
\end{align*}
According to Theorem \ref{them:kkpower} [cf. \cite{kkpower}]   
\begin{equation}
\label{eq:estimate1}
D^{(a)}=O(n^{1-\frac{a}{2}}),
\end{equation}
when $a>0.$

Firstly, we estimate $E_{(1-6)}^{(a)},$  when  $a\ge 1.$ % Let $i=(j-1)\sqrt{n}+k,$ where $k\in\{1,2,\dots ,\sqrt{n}\}.$
Notice that
\begin{equation}
 \label{eq:minkow}
|y+z|^a\le(|y|+|z|)^a \le 2^{a-1}(|y|^a+|z|^a)\,\,\, \text{for}\,\,\, a\ge 1, y,z\in R.
\end{equation}
This inequality is the consequence of the fact that $f(x)=x^a$ is convex over $R_+$ for $a\ge 1.$
Using Inequality (\ref{eq:minkow}) for $y=x-\left(\frac{i}{n}-\frac{1}{2n}\right)$ and 
$z=\left(\frac{i}{n}-\frac{1}{2n}\right)-\left(\frac{j}{n^{1/d}}-\frac{1}{2n^{1/d}}\right)$ 
we get
\begin{align}
%$$
& \notag
\left|x-\left(\frac{j}{n^{1/d}}-\frac{1}{2n^{1/d}}\right)\right|^a \\
%$$
&\leq \label{eq:estimate2}
2^{a-1}\left(\left|x-\left(\frac{i}{n}-\frac{1}{2n}\right)\right|^a+\left|\frac{i}{n}-\frac{j}{n^{1/d}}+\frac{1}{2n^{1/d}}-\frac{1}{2n}\right|^a\right).
\end{align}
We apply the definition of the Beta function (see Equation~(\ref{eq:beta2})) with parameters $i, n-i+1$, as well as Equation (\ref{Emult}) to deduce that
$$
\sum_{j=1}^{n^{1/d}}\sum_{i=(j-1)n^{(d-1)/d}+1}^{jn^{(d-1)/d}}
i\binom{n}{i}\int_0^1\left|\frac{i}{n}-\frac{j}{n^{1/d}}+\frac{1}{2n^{1/d}}-\frac{1}{2n}\right|^ax^{i-1}(1-x)^{n-i}dx
$$
$$
\sum_{j=1}^{n^{1/d}}\sum_{k=1}^{n^{(d-1)/d}}
\left|\frac{(j-1)n^{(d-1)/d}+k}{n}-\frac{j}{n^{1/d}}+\frac{1}{2n^{1/d}}-\frac{1}{2n}\right|^a
$$
\begin{equation}
\label{eq:estt}
=\sum_{j=1}^{n^{1/d}}\sum_{k=1}^{n^{(d-1)/d}}
\left|\frac{k}{n}-\frac{1}{2n^{1/d}}-\frac{1}{2n}\right|^a
\le\sum_{j=1}^{n^{1/d}}\sum_{k=1}^{n^{(d-1)/d}}
\left(\frac{1}{2n^{1/d}}\right)^a=\frac{n^{1-\frac{a}{d}}}{2^a}.
\end{equation}
Putting together Formulas (\ref{eq:estimate1}), (\ref{eq:estimate2}), and (\ref{eq:estt}) we obtain
\begin{equation}
\label{eq:estimate3}
E_{(1-6)}^{(a)}=O(n^{1-\frac{a}{2}})+O(n^{1-\frac{a}{d}})=O(n^{1-\frac{a}{d}}),\,\,\,\text{when}\,\,\, a\ge 1, d\ge 2
\end{equation}
To estimate $E_{(1-6)}^{(a)},$ when $0<a<1$ we define
$$F_{(i,j)}^{(a)}=i\binom{n}{i}\int_{0}^{1}\left|x-\left(\frac{j}{n^{1/d}}-\frac{1}{2n^{1/d}}\right)\right|^a
x^{i-1}(1-x)^{n-i}dx.$$
Observe that,
$$
E_{(1-6)}^{(a)}=\sum_{j=1}^{n^{1/d}}\sum_{i=(j-1)n^{(d-1)/d}+1}^{jn^{(d-1)/d}}F_{(i,j)}^{(a)}.
$$
Then, we use the discrete H\"older inequality 
with parameters $\frac{1}{a}$ and $\frac{1}{1-a}$
to derive
$$
\sum_{j=1}^{n^{1/d}}\sum_{i=(j-1)n^{\frac{d-1}{d}}+1}^{jn^{\frac{d-1}{d}}}
F^{(a)}_{(i,j)} \le
$$
$$
%&\le& \notag
\left(\sum_{j=1}^{n^{1/d}}\sum_{i=(j-1)n^{\frac{d-1}{d}}+1}^{jn^{\frac{d-1}{d}}}\left(F^{(a)}_{(i,j)}\right)^{\frac{1}{a}}\right)^{\frac{a}{1}}
\left(\sum_{j=1}^{n^{1/d}}\sum_{i=(j-1)n^{\frac{d-1}{d}}+1}^{jn^{\frac{d-1}{d}}}1\right)^{1-a}
$$
\begin{equation}
\label{eq:holder2}
\left(\sum_{j=1}^{n^{1/d}}\sum_{i=(j-1)n^{(d-1)/d}+1}^{jn^{(d-1)/d}}\left(F^{(a)}_{(i,j)}\right)^{\frac{1}{a}}\right)^{a}
(n)^{1-a}.
\end{equation}
Next, we use H\"older inequality for integrals with parameters $\frac{1}{a}$ and $\frac{1}{1-a}$ and get
$$\int_0^1\left|x-\left(\frac{j}{n^{1/d}}-\frac{1}{2n^{1/d}}\right)\right|^ax^{i-1}(1-x)^{n-i}i\binom{n}{i}dx\le
$$
$$
\left(\int_0^1\left(\left|x-\left(\frac{j}{n^{1/d}}-\frac{1}{2n^{1/d}}\right)\right|^a\right)^{\frac{1}{a}}x^{i-1}(1-x)^{n-i}i\binom{n}{i}dx\right)^{\frac{a}{1}},
$$
so $\left(F^{(a)}_{(i,j)}\right)^{\frac{1}{a}}\le F^{(1)}_{(i,j)}.$
Putting together Equation~\eqref{eq:estimate3} and Equation~\eqref{eq:holder2} we obtain
\begin{equation}
\label{eq:estimate4}
E_{(1-6)}^{(a)}=O\left(n^{1-\frac{a}{d}}\right),\,\,\, \text{when}\,\,\,0<a<1.
\end{equation}
%Next we use H\"older inequality and conclude that the total displacement  at the steps (1-5) equals
%$O(\sqrt{n}).$

Observe that in step $(8)$ of algorithm $MV_d(n,1)$
we have that $n^{\frac{d-1}{d}}$ mobile sensors are 
thrown uniformly and indendently at random in the unit $(d-1)$-dimensional cube.
According to inductive assumption the  
expected $a-$total movement at the step (8) is equal $O((n^{\frac{d-1}{d}})^{1-\frac{a}{2(d-1)}}).$
Hence the expected $a-$total movement in steps $(7-9)$
is in $O(n^{1/d}(n^{\frac{d-1}{d}})^{1-\frac{a}{2(d-1)}})=O(n^{1-\frac{a}{2d}}).$ 
Notice that the  
expected $a-$total movement in steps (1-6) is equal $O(n^{1-\frac{a}{d}})$ (see Formula (\ref{eq:estimate3})
and  Formula (\ref{eq:estimate4})).
Therefore, the expected cost of displacement to power $a$ of algorithm $MV_d(n,1)$ is in $O(n^{1-\frac{a}{2d}})$. This gives the claimed estimate
for $d$ and completes the proof of Theorem
\ref{thm:3_4}. %\qed
\end{proof}

Now we study a lower bound on the total displacement, when the $d-$dimensional cube sensing radius of the sensors is larger than $\frac{1}{2n^{1/d}}$.
First, we give a lemma which indicates how to scale the results of Theorem \ref{thm:3_4}
to $d-$dimensional cube of arbitrary length. The following lemma states that  Algorithm $MV_d(n,y)$ uses 
$O\left(y^an^{1-\frac{a}{2d}}\right)$
expected $a-$total movement.

\begin{Lemma}
\label{lem:scale}
Fix $d\in N.$
Let $n=m^d$ for some $m\in N$ and  let $a>0$.
Assume that $n$ sensors of d-dimensional cube sensing 
radius equal to $\frac{y}{2n^{1/d}}$ are thrown randomly and uniformly and 
independently with the uniform distribution on the $[0,y]^d.$ 
The expected $a-$total movement of algorithm $MV_d(n,y)$
is in $O\left(y^an^{1-\frac{a}{2d}}\right).$
\end{Lemma}
\begin{proof}
Assume that, $n$ sensors are in the cube $[0,y]^d.$ Then, multiply their coordinates by $1/y.$ From Theorem \ref{thm:3_4} 
the expected $a-$total movement in the unit cube $[0,1]^d$ is in $O\left(n^{1-\frac{a}{2d}}\right).$ Now by multiplying their coordinates by $y$
we get the result in the statement of the lemma. 
\end{proof}
A natural question to ask is: how to exploit the proposed Algorithm $MV_d(n,1)$ %and Algorithm \ref{alg_3_4_scale}
when the number
$n$ of nodes is not a $d-$th power  of natural number. Assume that  $n$ sensors have the $d-$dimensional cube sensing 
radius $r=\frac{f}{2n^{1/d}}$ and $f\ge\frac{n^{1/d}}{\lfloor n^{1/d}\rfloor}.$ 
To attain coverage of the cube $[0,1]^d$ choose $\lfloor n^{1/d}\rfloor^d$ sensors at  random and use Algorithm $MV_d(n,1)$ for the
chosen sensors. Then similar arguments hold for Algorithm $MV_d(n,y).$

Notice that for $f\ge \frac{3^{3/d}}{(3^{1/d}-1)(3^{1/d}-1)}$ we can do better.
The following theorem states that Algorithm $LV_d(n)$ uses $O\left(n^{1-\frac{a}{2d}}\left(\frac{\ln n}{n}\right)^{\frac{a}{2d}}\right)$
expected $a-$total movement. 

\begin{algorithm}
\caption{$LV_d(n)\,\,$ Sensor displacement on a unit $d$-dimensional cube  when $d\ge 2,$
\,\,\,\, $p=\frac{9}{4}\left(2+\frac{a}{d}\right),\,\,$ $A=\frac{3}{4}\left(2+\frac{a}{d}\right),$
$x_0$ is the real solution of the equation $\frac{x}{\frac{9}{4}(2+\frac{a}{d})\ln x}=3$ such that $x_0\ge 3$
}
\label{alg_log}
\begin{algorithmic}[1]
 \REQUIRE $n\ge \lceil x_0 \rceil$ mobile sensors with identical square sensing radius $r=\frac{f}{2 n^{1/d}},\,\,$ 
 %$f^2>\frac{n}{\lfloor\sqrt{n}\rfloor^2}$ 
 $f\ge \frac{3^{3/d}}{(3^{1/d}-1)(3^{1/d}-1)}$
 placed uniformly and independently at random on the cube $[0,1]^d.$
 \ENSURE  The final positions of sensors to attain coverage of the cube $[0,1]^d$
 %\STATE Choose parameter $p$ and $A$ such that %$p>\left(\frac{1}{\sqrt{\frac{1}{2+2\left(1-\frac{a}{2d}\right)}}-\sqrt{\frac{\ln n}{n}}}\right)^2,$ 
 %$p=\frac{9}{4}\left(2+\frac{a}{d}\right),\,\,$ $A=\frac{3}{4}\left(2+\frac{a}{d}\right).$
 %$\left\lceil\sqrt{\frac{p\ln n}{f^2}}\,\right\rceil^2\le\frac{n}{\left\lceil\sqrt{\frac{n}{p\ln n}}\right\rceil^2}-\sqrt{\frac{\left(2+2\left(1-\frac{a}{2d}\right)\right)n\ln n }{\left\lceil\sqrt{\frac{n}{p\ln n}}\right\rceil^2}}$
 %$f^2>\frac{p}{p-\sqrt{\left(2+2\left(1-\frac{a}{2d}\right)\right)p}}$ and $\frac{p\ln n}{f^2}=k^2$   for some $k\in N.$
  \STATE Divide the $d-$dimensional unit cube into $d-$dimensional subcubes of side $\frac{1}{\left\lfloor\left(\frac{n}{p\ln n}\right)^{1/d}\right\rfloor}$;
  \IF{there is a $d$-dimensional subcube with fewer than 
  $\frac{1}{3}\frac{n}{\left\lfloor\left(\frac{n}{p\ln n}\right)^{1/d}\right\rfloor^d}$
 sensors}
  \STATE{choose $\lfloor n^{1/d}\rfloor^d$ sensors at random;} 
  \STATE{use Algorithm $MV_d(n,1)$ that moves all
  $n:=\lfloor n^{1/d}\rfloor^d$ sensors to equidistant points that are sufficient to cover the $d$-dimensional subcube;}
  \ELSE 
   \STATE {In each $d$-dimensional subcube choose $\left\lfloor\left(A\ln n\right)^{1/d}\right\rfloor^d$ 
  %($\frac{p\ln n}{f(n)}<(p-\sqrt{(2+b)p})\ln n$ since $f(n)>\frac{p}{p-\sqrt{(2+b)p}}$) 
  sensors at random and use Algorithm $MV_d(n,y)$ with $n:=\left\lfloor\left(A\ln n\right)^{1/d}\right\rfloor^d, 
  y:=\frac{1}{\left\lfloor\left(\frac{n}{p\ln n}\right)^{1/d}\right\rfloor}$
  to move the chosen sensors to equidistant positions so as to cover the $d$-dimensional subcube;}
  \ENDIF
\end{algorithmic}
\end{algorithm}

\begin{theorem}
\label{thm:log}
Fix $d\in N\setminus\{1\}$ and $a>0.$
Let $f\ge \frac{3^{3/d}}{(3^{1/d}-1)(3^{1/d}-1)}$ and  $n\ge \lceil x_0 \rceil,$ where $x_0$
is the real solution of the equation $\frac{x}{\frac{9}{4}(2+\frac{a}{d})\ln x}=3$ such that $x_0\ge 3.$
Assume that $n$ sensors of d-dimensional cube sensing 
radius $r=\frac{f}{2n^{1/d}}$ are thrown randomly and uniformly and 
independently with the uniform distribution on the $[0,1]^d.$ 
The expected $a-$total movement of algorithm $LV_d(n)$
is in 
$O\left(n^{1-\frac{a}{2d}}\left(\frac{\ln n}{n}\right)^{\frac{a}{2d}}\right).$
\end{theorem} 
\begin{proof}
Assume that $d\in N\setminus\{1\}$ and $a>0.$
Let $p=\frac{9}{4}\left(2+\frac{a}{d}\right)$ and  $A=\frac{3}{4}\left(2+\frac{a}{d}\right),$ 
$x_0$ is the real solution of the equation $\frac{x}{\frac{9}{4}(2+\frac{a}{d})\ln x}=3$ such that $x_0\ge 3.$
First of all, observe that $\frac{n}{p\ln (n)}> 3$ for $n\ge \lceil x_0\rceil .$
We will prove that Algorithm $LV_d(n)$ uses $O\left(n^{1-\frac{a}{2d}}\left(\frac{\ln n}{n}\right)^{\frac{a}{2d}}\right)$
expected $a-$total movement.
There are two cases to consider.

Case 1: There exists a $d-$dimensional subcube with fewer than
$$\frac{1}{3}\frac{n}{\left\lfloor\left(\frac{n}{p\ln n}\right)^{1/d}\right\rfloor^d}$$
sensors. In this case choose 
$\lfloor n^{1/d} \rfloor^d$ sensors uniformly and randomly from $n$ sensors. Applying the inequalities
$\lfloor x\rfloor > x-1$ and $f\ge \frac{3^{3/d}}{(3^{1/d}-1)(3^{1/d}-1)}>\frac{3^{1/d}}{3^{1/d}-1}$
we deduce that 
$$
\left(\lfloor n^{1/d}\rfloor\frac{f}{n^{1/d}}\right)^d>\left(\frac{n^{1/d}-1}{n^{1/d}}\frac{3^{1/d}}{3^{1/d}-1}\right)^d\ge 1\,\,\,\text{for}\,\,\, n\ge 3.
$$ 
Therefore, the $\lfloor n^{1/d}\rfloor^d$
chosen sensors are enough to attain the coverage.
The expected $a-$total movement is 
$O\left(\left(\lfloor n^{1/d}\rfloor^d\right)^{1-\frac{a}{2d}}\right)=O\left(n^{1-\frac{a}{2d}}\right)$ by Theorem \ref{thm:3_4}.

Case 2:  All $d-$dimensional subcubes contain at least $\frac{1}{3}\frac{n}{\left\lfloor\left(\frac{n}{p\ln n}\right)^{1/d}\right\rfloor^d}$ sensors. From the inequality 
$\lfloor x\rfloor \le x$
we deduce that, 
$$\left\lfloor\left(A\ln n\right)^{1/d}\right\rfloor^d\le\frac{1}{3}\frac{n}{\left\lfloor\left(\frac{n}{p\ln n}\right)^{1/d}\right\rfloor^d}.$$
Hence it is possible to choose $\left\lfloor\left(A\ln n\right)^{1/d}\right\rfloor^d$ sensors at random in each $d$-dimensional subcube
with more than 
$\frac{1}{3}\frac{n}{\left\lfloor\left(\frac{n}{p\ln n}\right)^{1/d}\right\rfloor^d}$
sensors. 
Let us consider the sequence

$$a_n=\frac{3^{3/d}}{(3^{1/d}-1)(3^{1/d}-1)}\left\lfloor\left(A\ln n\right)^{1/d}\right\rfloor \frac{1}{n^{1/d}} \left\lfloor\left(\frac{n}{p\ln n}\right)^{1/d}\right\rfloor
$$
for $n\ge \lceil x_0\rceil.$ 
%Numerical calculation for small values of 
%$n$ ($31 < n < 250$) confirms that $a_n>1.$
Applying inequality $\lfloor x\rfloor>x-1$ we see that
$$
a_n>\frac{3^{3/d}}{(3^{1/d}-1)(3^{1/d}-1)}\left(\left(A\ln n\right)^{1/d}-1\right) \frac{1}{n^{1/d}} \left(\left(\frac{n}{p\ln n}\right)^{1/d}-1\right)
$$
\begin{equation}
\label{eq:an1}
=\frac{3^{2/d}}{(3^{1/d}-1)(3^{1/d}-1)}\left(1-\frac{1}{\left(A\ln n\right)^{1/d}}\right)  \left(1-\left(\frac{p\ln n}{n}\right)^{1/d}\right)
\end{equation}
Observe that
\begin{equation}
\label{eq:an2}
\frac{p\ln n}{n} \le \frac{1}{3},\,\, \frac{1}{A\ln n}\le \frac{1}{3}\,\,\,\text{for}\,\,\, n\ge \lceil x_0\rceil
\end{equation}

Putting together Equation (\ref{eq:an1}) and Equation (\ref{eq:an2}) we get
$$\left\lfloor\left(A\ln n\right)^{1/d}\right\rfloor^d \frac{f^d}{n} \left\lfloor\left(\frac{n}{p\ln n}\right)^{1/d}\right\rfloor^d\ge a^d_n> 1.$$
Therefore, $\left\lfloor\left(A\ln n\right)^{1/d}\right\rfloor^d$ 
chosen sensors are enough to attain the coverage.
By the independence of the sensors positions, the $\left\lfloor\left(A\ln n\right)^{1/d}\right\rfloor^d$ 
chosen sensors in any given $d$-dimensional subcube are distributed randomly and independently with uniform distribution over the $d$-dimensional subcube of side
%$\sqrt{\frac{p\ln n}{n}}.$  
$y=\frac{1}{\left\lfloor\left(\frac{n}{p\ln n}\right)^{1/d}\right\rfloor}.$
By Lemma \ref{lem:scale} the expected $a-$total movement inside each $d$-dimensional subcube is 
$$O\left(\left(\frac{1}{\left\lfloor\left(\frac{n}{p\ln n}\right)^{1/d}\right\rfloor}\right)^a\left(\left\lfloor(A\ln n)^{1/d}\right\rfloor^d\right)^{1-\frac{a}{2d}}\right)=
O\left(\frac{(\ln n)^{\frac{a}{2d}}}{n^{\frac{a}{d}}}(\ln n)\right).$$
Since, there are $\left\lfloor\left(\frac{n}{p\ln n}\right)^{1/d}\right\rfloor^d$ $d$-dimensional subcubes, the expected $a-$total movement  over all $d$-dimensional subcubes must be in
$O\left(n^{1-\frac{a}{2d}}\left(\frac{\ln n}{n}\right)^{\frac{a}{2d}}\right).$
It remains to consider the probability with which each of these cases occurs. The proof of the theorem will be a consequence of the following Claim.
\begin{Claim}
\label{claim:first}
Let $p=\frac{9}{4}\left(2+\frac{a}{d}\right).$
 The probability that fewer than 
 $\frac{1}{3}\frac{n}{\left\lfloor\left(\frac{n}{p\ln n}\right)^{1/d}\right\rfloor^d}$ 
 %$\frac{1}{2}\frac{n}{\left\lfloor\left(\frac{n}{p\ln n}\right)^{1/d}\right\rfloor^d}-\sqrt{\frac{\left(2+2\left(1-\frac{a}{2d}\right)\right)n\ln n }{\left\lfloor\left(\frac{n}{p\ln n}\right)^{1/d}\right\rfloor^d}}$ 
 sensors fall in any $d$-dimensional subcube is
 $<\frac{\left\lfloor\left(\frac{n}{p\ln n}\right)^{1/d}\right\rfloor^d}{n^{1+\frac{a}{2d}}}.$
\end{Claim}
\begin{proof} (Claim~\ref{claim:first})
First of all, from the inequality $\lfloor x\rfloor \le x$ we get
$$\sqrt{\frac{\left(2+\frac{a}{d}\right)\ln n}{n}\left\lfloor\left(\frac{n}{p\ln n}\right)^{1/d}\right\rfloor^d}\le \frac{2}{3}.$$ 
Hence,
\begin{equation}
\label{eq:chernof}
\frac{1}{3}\frac{n}{\left\lfloor\left(\frac{n}{p\ln n}\right)^{1/d}\right\rfloor^d}\le
\frac{n}{\left\lfloor\left(\frac{n}{p\ln n}\right)^{1/d}\right\rfloor^d}-\sqrt{\frac{\left(2+\frac{a}{d}\right)n \ln n  }{\left\lfloor\left(\frac{n}{p\ln n}\right)^{1/d}\right\rfloor^d}}.
\end{equation}
The number of sensors falling in a $d$-dimensional subcube is a Bernoulli process with probability of success $\frac{1}{\left\lfloor\left(\frac{n}{p\ln n}\right)^{1/d}\right\rfloor^d}.$ By Chernoff bounds, the probability
that a given $d$-dimensional subcube has fewer than 
$$
\frac{n}{\left\lfloor\left(\frac{n}{p\ln n}\right)^{1/d}\right\rfloor^d}-\sqrt{\frac{\left(2+\frac{a}{d}\right)n\ln n }{\left\lfloor\left(\frac{n}{p\ln n}\right)^{1/d}\right\rfloor^d}}
$$ 
sensors is less than 
$e^{-\left(1+\frac{a}{2d}\right)\ln n}<\frac{1}{n^{1+\frac{a}{2d}}}.$ 
Specifically we use the Chernoff bound 
$$\Pr[X<(1-\delta)m]<e^{-{\delta}^2m/2},$$
$m=\frac{n}{\left\lfloor\left(\frac{n}{p\ln n}\right)^{1/d}\right\rfloor^d},$
$\delta=\sqrt{\frac{\left(2+\frac{a}{d}\right)\ln n}{n}\left\lfloor\left(\frac{n}{p\ln n}\right)^{1/d}\right\rfloor^{d}}.$
As there are $\left\lfloor\left(\frac{n}{p\ln n}\right)^{1/d}\right\rfloor^d$ $d$-dimensional subcubes, the event that one has fewer than 
$$
\frac{n}{\left\lfloor\left(\frac{n}{p\ln n}\right)^{1/d}\right\rfloor^d}-\sqrt{\frac{\left(2+\frac{a}{d}\right)n\ln n }{\left\lfloor\left(\frac{n}{p\ln n}\right)^{1/d}\right\rfloor^d}}.
$$
sensors occurs with probability less than $\frac{\left\lfloor\left(\frac{n}{p\ln n}\right)^{1/d}\right\rfloor^d}{n^{1+\frac{a}{2d}}}.$ This and Equation \eqref{eq:chernof} completes the proof of Claim \ref{claim:first}. \qed
\end{proof}
Using Claim  \ref{claim:first} we can upper bound the expected $a-$total movement as follows:
$$\left(1-\frac{\left\lfloor\left(\frac{n}{p\ln n}\right)^{1/d}\right\rfloor^{d}}{n^{1+\frac{a}{2d}}}\right)
O\left(n^{1-\frac{a}{2d}}\left(\frac{\ln n}{n}\right)^{\frac{a}{2d}}\right)
+\left(\frac{\left\lfloor\left(\frac{n}{p\ln n}\right)^{1/d}\right\rfloor^{d}}{n^{1+\frac{a}{2d}}}\right)O\left(n^{1-\frac{a}{2d}}\right)=$$
$$O\left(n^{1-\frac{a}{2d}}\left(\frac{\ln n}{n}\right)^{\frac{a}{2d}}\right),$$
which proves Theorem \ref{thm:log}. \qed
\end{proof}

\section{Simulation Results}
\label{sec:simu}

In this section we use simulation results to analyze how random placement of sensors on the square impacts the expected
$a-$total movement. 

We repeated 3 times the following experiments.
Firstly, for each number of sensors $n\in\{2^2,3^2,4^2,\dots,60^2\}$ we generated $32$ random placements.
Then we calculated the expected $a-$total movement according to Algorithm $MV_2(n,x)$.
Let $E_{n,32}$ be the average of $32$ measurements of the expected $a-$total movement. Then, we placed the points in the set $\{(n, E_{n,32}): n=2^2,3^2,4^2,\dots,60^2\}$ into the picture. 

\begin{figure}
  \begin{center}
\begin{tabular}{cc}
 \begin{minipage}[b]{0.50\textwidth} \centering \includegraphics[width=0.90\textwidth]{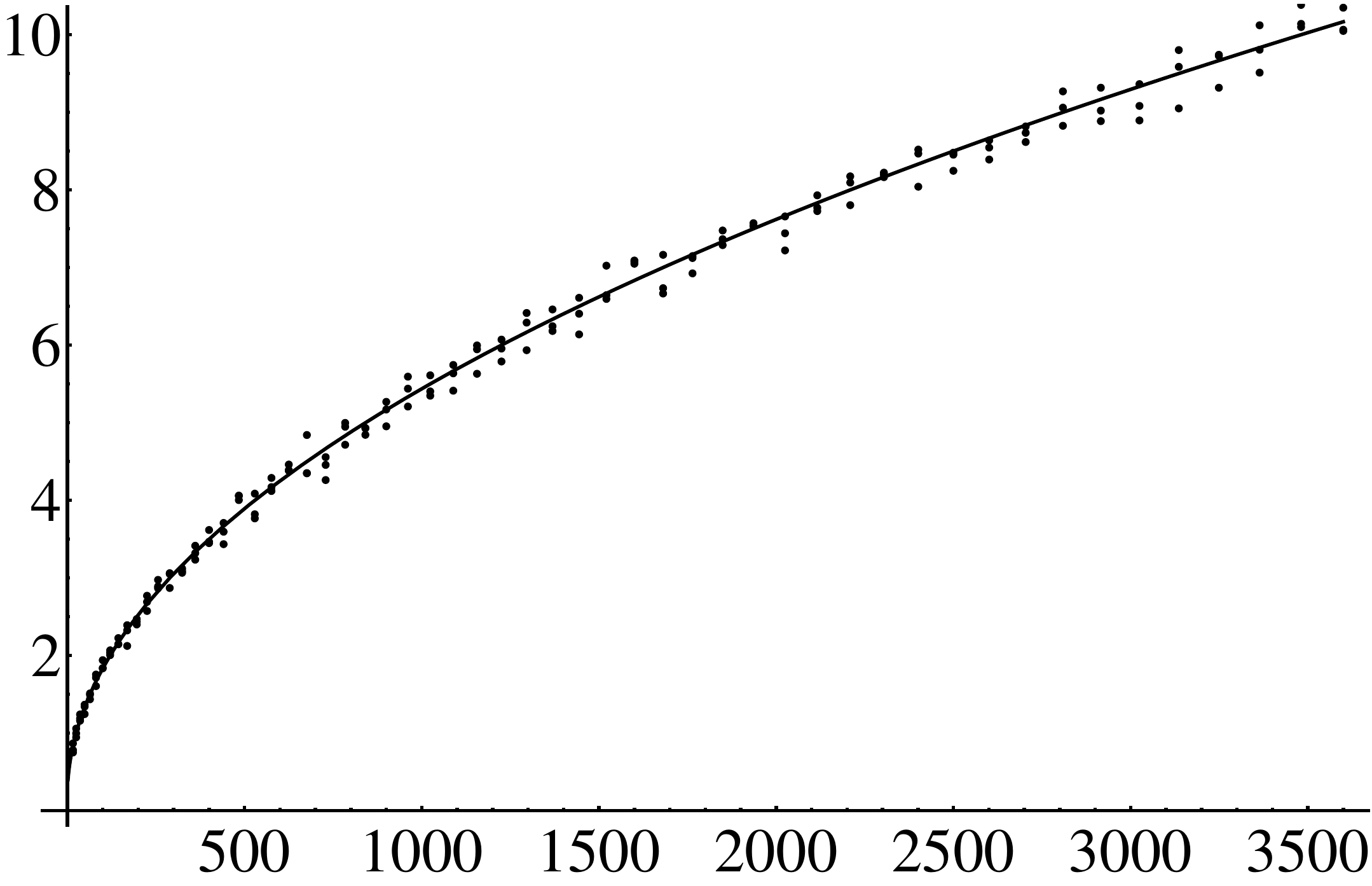}\\ $a=2, M_{2}=\frac{1}{6}\sqrt{n}$ \end{minipage}
&
 \begin{minipage}[b]{0.50\textwidth} \centering \includegraphics[width=0.90\textwidth]{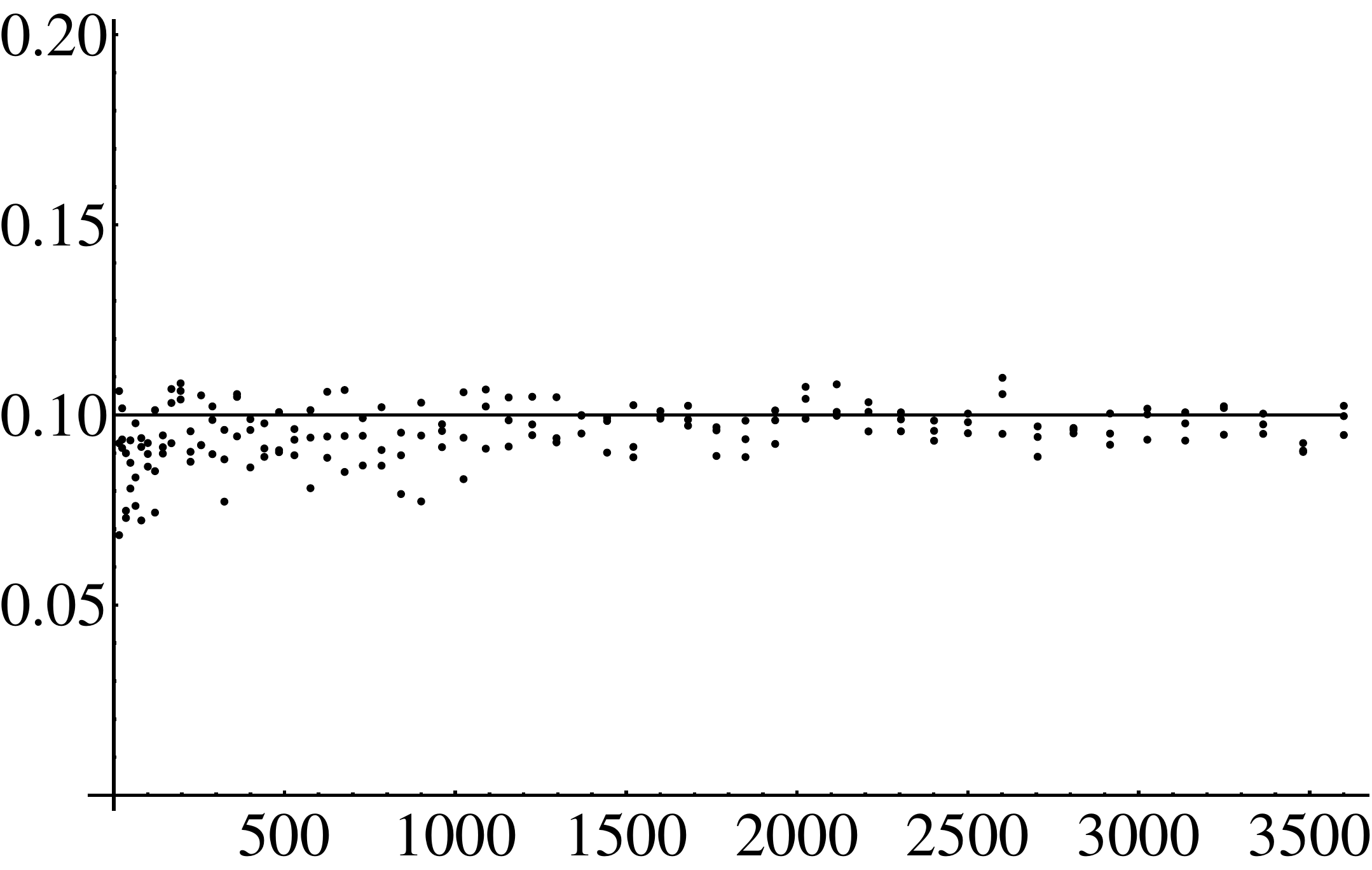}\\ $a=4, M_{4}=\frac{1}{10}$ \end{minipage}
  \\
\end{tabular}
\end{center}
  \caption{The expected $a-$total movement of Algorithm $MV_2(n,1).$}
    \label{fig:alg1}
\end{figure}
Figure \ref{fig:alg1} illustrates the described experiments for Algorithm $MV_2(n,1)$  when $a=2$ and $a=4.$
The additional line in the above pictures is the plot of the function which is the theoretical estimation.
Black dots which represent numerical results are situated near the theoretical line.
According to the proof of Theorem \ref{thm:3_4} the steps (7-9) of Algorithm $MV_2(n,1)$ conctribute the asymptotics.
Notice that, the expected $a-$total movement in steps (7-9) of Algorithm $MV_2(n,1)$ is equal to
$$E_{(7-9)}^{(a)}=\sqrt{n}\sum_{i=1}^{\sqrt{n}}i\binom{\sqrt{n}}{i}\int_{0}^{1}\left|x-\left(\frac{j}{\sqrt{n}}-\frac{1}{2\sqrt{n}}\right)\right|^a
x^{i-1}(1-x)^{\sqrt{n}-i}dx.$$
Applying the Formulas for $E_{(7-9)}^{(2)}$ and $E_{(7-9)}^{(4)}$ in any mathematical software that performs symbolic calculation we get
$$E_{(7-9)}^{(2)}\sim \frac{1}{6}\sqrt{n} \mbox{ and } E_{(7-9)}^{(4)}\sim \frac{1}{10}.$$ Therefore, $M_{2}=\frac{1}{6}\sqrt{n}$ and
$M_{4}=\frac{1}{10}.$

\begin{figure}
  \begin{center}
\begin{tabular}{cc}
 \begin{minipage}[b]{0.50\textwidth} \centering \includegraphics[width=0.90\textwidth]{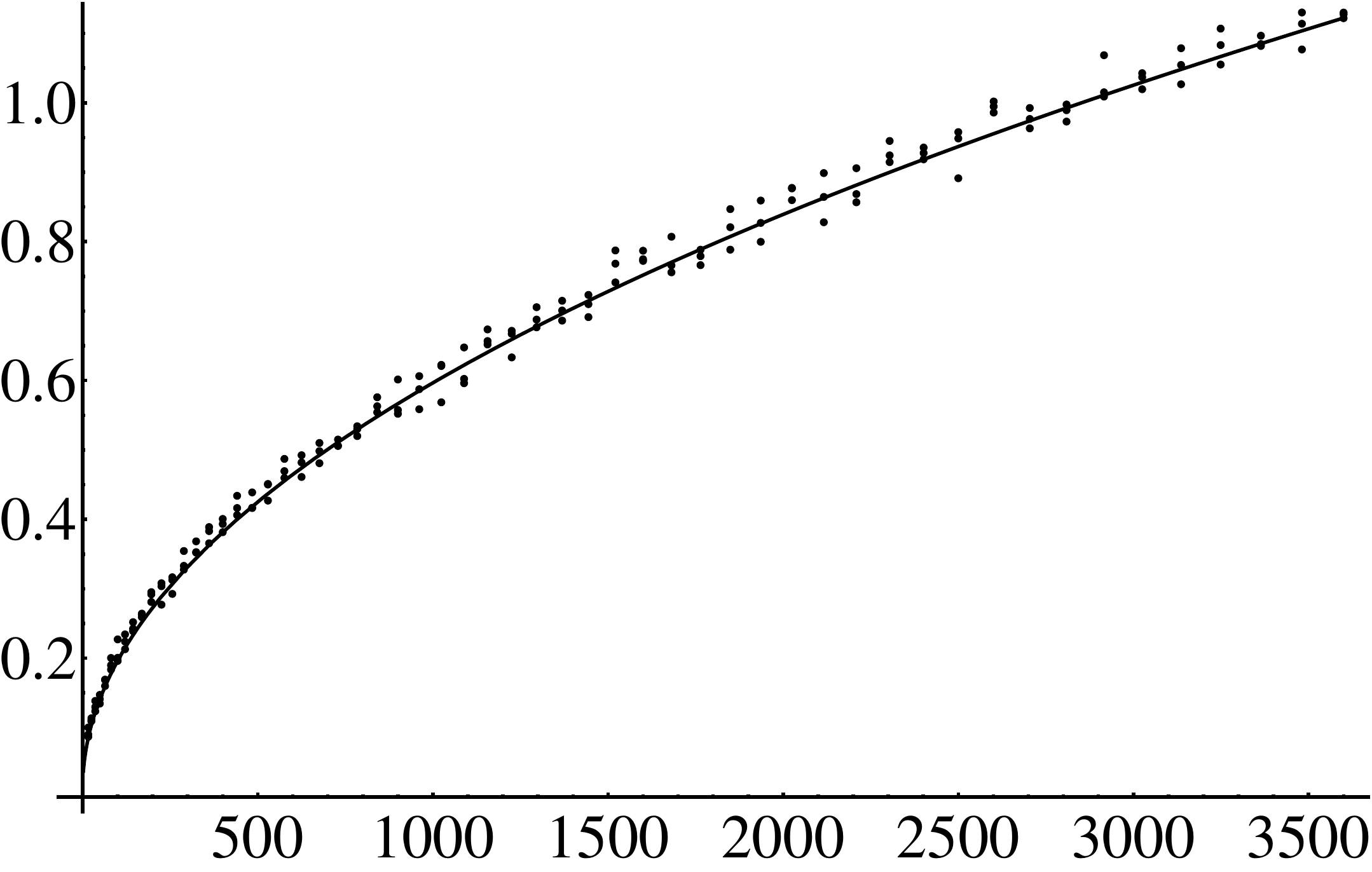}\\ $a=2, y=\frac{1}{3}, M_{2}=\frac{1}{54}\sqrt{n}$ \end{minipage}
&
 \begin{minipage}[b]{0.50\textwidth} \centering \includegraphics[width=0.90\textwidth]{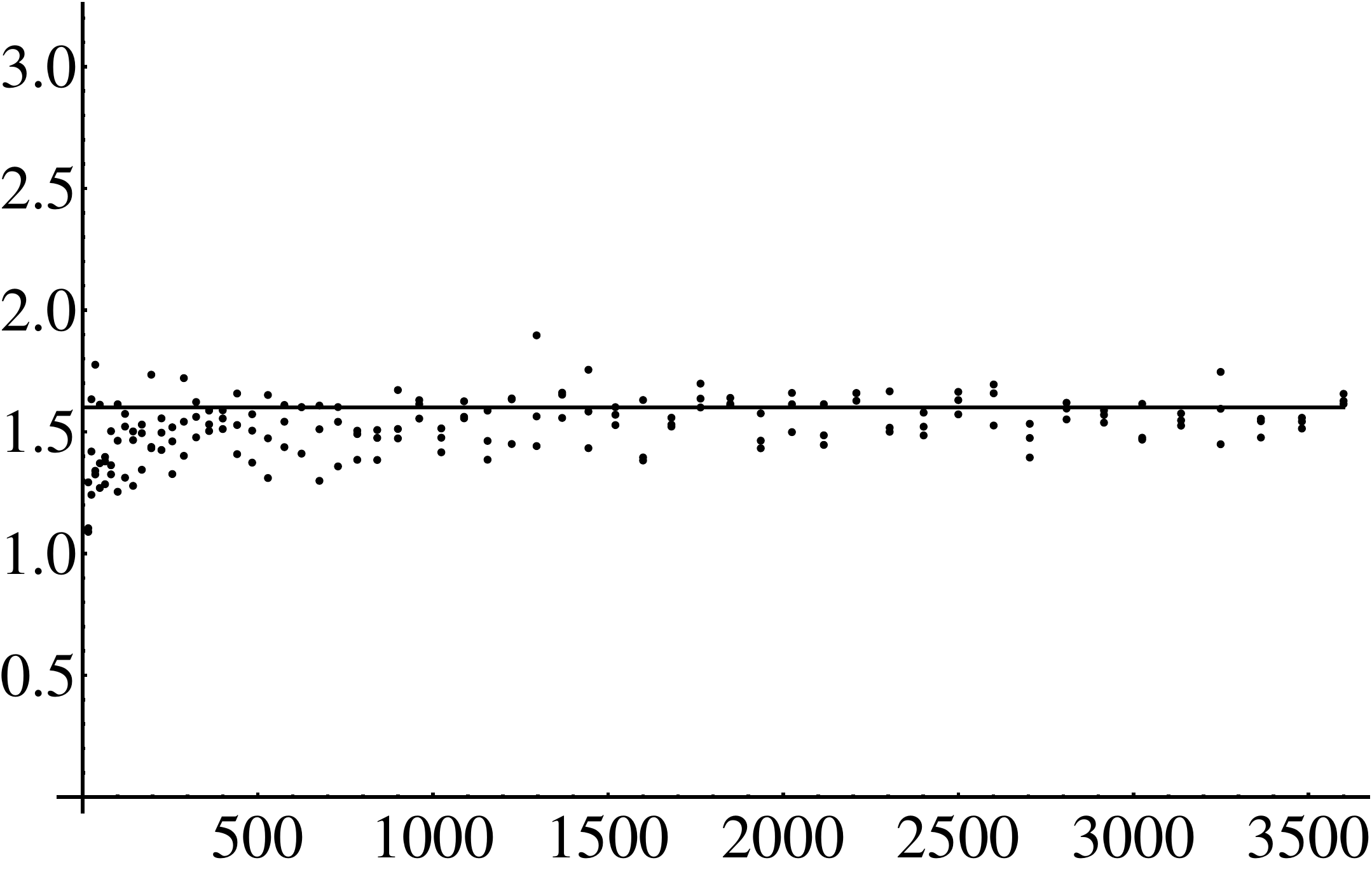}\\ $a=4, y=2, M_{4}=\frac{4}{5}$ \end{minipage}
  \\
\end{tabular}
  \end{center}
  \caption{The expected $a-$total movement of Algorithm $MV_2(n,y)$.}
    \label{fig:alg2}
\end{figure}

Figure \ref{fig:alg2} illustrates the described experiments for Algorithm $MV_2(n,\frac{1}{3})$ when $a=2$ and $MV_2(n,2)$ when $a=4.$
The additional line in the above pictures is the plot of the function which is the theoretical estimation.
Black dots which represent numerical results are situated near the theoretical line.
According to the proof of Lemma \ref{lem:scale} we have $M_{2}=\left(\frac{1}{3}\right)^2\frac{1}{6}\sqrt{n}=\frac{1}{54}\sqrt{n}$
and $M_{4}=2^4\frac{1}{10}=\frac{4}{5}.$

\section{Conclusion}
\label{sec:conclusion}

In this paper we studied the movement of $n$ sensors with identical square sensing radius in $d$ dimensions when the cost of movement
of sensor is proportional to some (fixed) power $a>0$ of the distance traveled. We obtained bounds on the movement depending
on the range of sensors. %An interesting problem would be to study the displacement to a power $a$ for all $a>0$ in higher dimensions.
%\bibliographystyle{plain}
%\bibliography{refs,refs1}

\bibliographystyle{plain}
\bibliography{refs,refs1}

\end{document}